\documentclass[conference]{IEEEtran}

\usepackage{amsmath,amssymb,amsthm}
\usepackage{cite}
\usepackage{epsfig}
\usepackage[ruled,vlined,lined,linesnumbered]{algorithm2e}
\usepackage{color}

\ifCLASSINFOpdf
\else
\fi
\newcommand{\defeq}{\stackrel{\text{def}}{=}}
\theoremstyle{definition}
\newtheorem{PROP}{Proposition}
\begin{document}
%
\title{Symbol-Based Successive Cancellation List Decoder for Polar Codes}

\IEEEoverridecommandlockouts

\author{
\IEEEauthorblockN{Chenrong Xiong, Jun Lin and Zhiyuan Yan}

\IEEEauthorblockA{Department of Electrical and Computer Engineering, Lehigh
  University, Bethlehem, PA 18015 USA\\
Email: \{chx310, jul311, zhy6\}@lehigh.edu}

}


%


\maketitle

\begin{abstract}
Polar codes is promising because they can provably achieve the channel capacity
while having an explicit construction method. Lots of work have been done for
the bit-based decoding algorithm for polar codes. In this paper, generalized
symbol-based successive cancellation (SC) and SC list decoding algorithms are
discussed. A symbol-based recursive channel combination relationship is proposed
to calculate the symbol-based channel transition probability. This proposed method
needs less additions than the maximum-likelihood decoder used by the existing
symbol-based polar decoding algorithm. In addition, a two-stage list pruning
network is proposed to simplify the list pruning network for the symbol-based SC
list decoding algorithm.
\end{abstract}


%
\IEEEpeerreviewmaketitle

\section{Introduction}
\label{sec:intro}
Since polar codes were introduced by Arikan \cite{5075875}, they have attracted lots
of interest in the fields of communication and coding theory, because they
can provably achieve the channel capacity not only for arbitrary discrete memoryless channels, but
also for any continuous memoryless channel \cite{5351487}. 
However, their capacity
approaching can be achieved only when the code length is large enough ($N >
2^{20}$ \cite{6327689}) under the SC decoding algorithm. For short or moderate
code length, in terms of the error performance, polar codes with the SC decoding
algorithm is worse than turbo codes or low-density parity-check codes
\cite{6297420, 6033837}.

To improve the error performance of polar codes, lots of work have been
done. Systematic polar codes \cite{5934670} was proposed to reduce the bit error
rate while guaranteeing the same frame error rate (FER) compared with their
non-systematic counterparts. 
An SC list decoding algorithm for polar codes was proposed in \cite{6033904}. The SC list decoding algorithm
outperforms the SC decoding algorithm and achieves the error performance close to that of
the ML decoding algorithm at the cost of complexity of $\mathcal{O}(LN\log N)$,
where $L$ is the list size. Moreover, the concatenation of polar codes with
cyclic redundancy check (CRC) codes was introduced in \cite{6297420, Tal2012}. To decode
the CRC-concatenated polar codes, a CRC detector is used in the SCL decoding
algorithm to help the codeword determination. The combination of an SCL decoding
algorithm and a CRC detector is called CRC-aided SCL (CA-SCL) decoding
algorithm. \cite{Tal2012} shows that with the CA-SCL decoding algorithm, the
error performance of a (2048, 1024) CRC-concatenated polar code is better that
of a (2304, 1152) LDPC code, which is used in the WiMax standard \cite{1603394}.

To implement decoders for polar codes, several works have been done for the SC
decoding algorithm. Arikan \cite{5075875} showed that a fully parallel SC
decoder has a latency of $2N-1$ clock cycles. This decoder has complexity of
$\mathcal{O}(N\log N)$. A tree SC decoder and a line SC decoder with complexity
of  $\mathcal{O}(N)$ were proposed in \cite{5946819}. These two decoders have
the same latency as the fully parallel SC decoder. To reduce complexity
further, Leroux \cite{6327689} proposed a semi-parallel SC decoder for polar
codes by taking advantage of the recursive structure of polar codes to reuse
processing resources. 
To reduce the latency, a simplified SC (SSC) polar decoder
was introduced in \cite{6065237} and it was further analyzed in
\cite{6680761}. In the SSC polar decoder, a polar code is converted to a binary
tree including three types of nodes: rate-one, rate-zero and rate-$R$
nodes. Based on the SSC polar decoder, the ML SSC decoder
makes use of the ML decoding algorithm to deal with 
rate-$R$ nodes in \cite{6464502, 6804939}. However, SSC and ML-SSC polar decoders
depend on positions of information bits and frozen bits, and are code-specific consequently. 
In \cite{6475198}, a pre-computation look-ahead technique was
proposed to help the tree SC decoder shorten the latency by half. An efficient SCL decoder
architecture was proposed in\cite{ListPolarJun1}.
Recently, parallel decoders of polar codes were proposed in\cite{ParSC}. To
avoid ambiguity between the aforementioned fully parallel SC decoder in
\cite{5075875} and parallel decoders in \cite{ParSC}, we call the latter as
symbol-based polar decoders in this paper because an $M$-bit symbol-based
polar decoder decodes $M$ bits at a time instead of only one bit. However,
\cite{ParSC} is focused on some specific case and does not provide a general
discussion. Meanwhile, it uses the ML decoder to calculate the symbol-based
channel transition probability, which is not complexity-efficient enough.

The main contributions of this paper are:
\begin{itemize}
\item Generalized symbol-based polar decoding algorithms are
  discussed. Furthermore, a symbol-based recursive channel combination relationship is derived
  to calculate the symbol-based channel transition
  probability. The proposed method needs less additions than the ML detector used in \cite{ParSC}.
\item An $M$-bit symbol-based SCL polar decoder needs to find $L$ most-reliable
  lists among $2^ML$ list candidates. A two-stage list pruning network are
  proposed to perform this list pruning function. $2^ML$ list candidates are
  divided into $L$ groups. Each group has $2^M$ list candidates. In the first
  stage, $q$ most-reliable lists for each group are found. Then, $L$
  most-reliable list candidates are sorted out from $qL$ list candidates
  generated by the first stage. If $q<L$, the two-stage list pruning network can
  achieve lower complexity and a shorter critical path delay than the list
  pruning network with $q=L$.
\end{itemize}

The rest of our paper is organized as follows. Section~\ref{sec:review} briefly reviews
polar codes and existing decoding algorithms. In Section~\ref{sec:SBDecoder},
the generalized $M$-bit symbol-based SC and SCL decoding algorithms for polar
codes are discussed. Based on the Arikan's recursive channel transformations, we
derive the symbol-based recursive channel combination relationship to calculate the
symbol-based channel transition probability. To simplify the selection of the
list candidates, a two-stage list pruning network is proposed in
Section~\ref{sec:TSLPN}. Some conclusions are given in Section~\ref{sec:conclusion}.

\section{Polar Codes and Existing Decoding Algorithms}
\label{sec:review}
\subsection{Polar Codes}

Polar codes are linear block codes. The block length of polar codes is
restricted to a power of two, $N=2^n$ for $n \geq 2$. We follow the notation for vectors in \cite{5075875}, namely $u_a^b=(u_a,u_{a+1},\cdots,u_{b-1},u_b)=(u_a^{b-1},u_b)$
. Assume $\mathbf{u}=u_0^{N-1}=(u_0,u_1,\cdots,u_{N-1})$ is the encoding bit
sequence. Let 
$F=\left[
\begin{array}{ccc}
1 & 0 \\
1 & 1 
\end{array}
\right]$. The corresponding encoded bit sequence
$\mathbf{x}=x_0^{N-1}=(x_0,x_1,\cdots,x_{N-1})$ is generated by
\begin{equation}
\mathbf{x} = \mathbf{u}B_NF^{\otimes n},
\label{equ:encoder}
\end{equation}
where $B_N$ is an $N\times N$ bit-reversal permutation matrix and $F^{\otimes
  n}$ denotes the $n$-th Kronecker power of $F$. 

For any index set $\mathcal{A} \subset \mathcal{I}=\{0,1,\cdots, N-1\}$, let
$\mathbf{u}_{\mathcal{A}}$ denote the sub-sequence of $\mathbf{u}$ defined by
$\mathbf{u}_{\mathcal{A}}=(u_i: i\in
\mathcal{A})$. Denote the complement of $\mathcal{A}$ in $\mathcal{I}$ as
$\mathcal{A}^c$. Let $\mathbf{u}_{\mathcal{A}^c}=(u_i:0 \leq i < N, i\notin
\mathcal{A})$. For an $(N,K)$ polar code, the encoding bit sequence is
grouped into two parts: a $K$-element part $\mathbf{u}_{\mathcal{A}}$ which
carries information bits, and $\mathbf{u}_{\mathcal{A}^c}$ whose elements are
predefined frozen bits. For the sake of convenience, frozen bits are set to
be zero. 

\subsection{SC Decoding Algorithm for Polar Codes}
Given a transmitted codeword $\mathbf{x}$ and the corresponding received word
$\mathbf{y}$, the SC decoding algorithm for an $(N,K)$ polar code decodes the
encoding bit sequence $\mathbf{u}$ from $u_0$ to $u_{N-1}$ successively one by
one as shown in Alg.~\ref{alg:SC}. Here, $\hat{\mathbf{u}} = (\hat{u}_0,
\hat{u}_1,\cdots,\hat{u}_{N-1})$  represents the estimated value for
$\mathbf{u}$. ${\rm P}(\mathbf{y},\hat{u}_0^{j-1}|u_j)$ is the
probability that $\mathbf{y}$ is received and the previously decoded bits are $\hat{u}_0^{j-1}$
given $u_j$ is zero or one. 

\begin{algorithm}
\caption{SC Decoding Algorithm \cite{5075875}}
\label{alg:SC}
\LinesNumbered
\For{$j=0:N-1$}{
\lIf{$j\in \mathcal{A}^c$}{
$\hat{u}_j=0$
}
\uElse{
\lIf{$\frac{{\rm P}(\mathbf{y},\hat{u}_0^{j-1}|1)}{{\rm P}(\mathbf{y},\hat{u}_0^{j-1}|0)} \geq 1 $}{
$\hat{u}_j = 1$
}
\lElse{
$\hat{u}_j = 0$
}
}
}
\end{algorithm}

To calculate ${\rm P}(\mathbf{y},\hat{u}_0^{j-1}|u_j)$, the following Arikan's
recursive channel transformations \cite{5075875} are used:
\begin{equation}
\begin{split}
\label{eq:ArikanT1}
{\rm
  P}&(y_0^{\Gamma-1},u_0^{2i-1}|u_{2i})\\
&=\frac{1}{2}\sum_{u_{2i+1}}\Bigl[{\rm
  P}(y_{0}^{{\Gamma}/2-1},u_{0,e}^{2i-1}\oplus u_{0,o}^{2i-1}|u_{2i}\oplus
u_{2i+1})\\
&\hspace{20mm}\cdot{\rm P}(y_{{\Gamma}/2}^{{\Gamma}-1},u_{0,o}^{2i-1}|u_{2i+1})\Bigr],
\end{split}
\end{equation}
and
\begin{equation}
\begin{split}
\label{eq:ArikanT2}
{\rm
  P}&(y_0^{\Gamma-1},u_0^{2i}|u_{2i+1})\\
&=\frac{1}{2}{\rm
  P}(y_{0}^{{\Gamma}/2-1},u_{0,e}^{2i-1}\oplus u_{0,o}^{2i-1}|u_{2i}\oplus
u_{2i+1})\\
&\hspace{20mm}\cdot{\rm P}(y_{{\Gamma}/2}^{{\Gamma}-1},u_{0,o}^{2i-1}|u_{2i+1}),
\end{split}
\end{equation}
where $1 \leq \Gamma=2^{\gamma} \leq N$, and $0 \leq i < \frac{N}{2}$. 

%

\subsection{SCL Decoding Algorithm for Polar Codes}
Instead of making decision for each information bit of $\mathbf{u}$ in an SC
decoding algorithm, the SCL
decoding algorithm \cite{6033904} creates two paths in which the bit is assumed to be 0 and
1, respectively. If the number of paths is greater than the list size $L$, the $L$ most-reliable
paths are selected out. At the end of the decoding procedure, the most reliable
path is chosen as $\hat{\mathbf{u}}$. The SCL decoding algorithm
is described in Alg. \ref{alg:SCL}. Without loss of generality, assume $L$ to
be a power of two, i.e.\ $L=2^l$. Let 
$\mathbf{L}_i=((\mathcal{L}_i)_0,(\mathcal{L}_i)_1,\cdots,(\mathcal{L}_i)_{N-1})$
represent the $i$-th list vector, where $0 \leq i < L$. 

\begin{algorithm}
\caption{SCL Decoding Algorithm \cite{6033904}}
\label{alg:SCL}
\LinesNumbered
$\alpha=1$\;
\For{$j=0:N-1$}{
\uIf{$j\in \mathcal{A}^c$}{
\For{$i=0:\alpha-1$}{
$(\mathcal{L}_i)_j = 0$\;
}
}
\uElseIf{$2\alpha\leq L$}{
\For{$i=0:\alpha-1$}{
$(\mathcal{L}_i)_0^j=((\mathcal{L}_i)_0^{j-1},0)$\;
$(\mathcal{L}_{i+\alpha})_0^j=((\mathcal{L}_i)_0^{j-1},1)$\;
}
$\alpha=2\alpha$\;
}
\Else{
\For{$i=0:L-1$}{
${\sf S}[i]{\sf .P}={\rm P}(\mathbf{y},(\mathcal{L}_i)_0^{j-1}|0)$\;
${\sf S}[i]{\sf .L}=(\mathcal{L}_i)_0^{j-1}$\;
${\sf S}[i]{\sf .U}=0$\;
${\sf S}[i+L]{\sf .P}={\rm P}(\mathbf{y},(\mathcal{L}_i)_0^{j-1}|1)$\;
${\sf S}[i+L]{\sf .L}=(\mathcal{L}_i)_0^{j-1}$\;
${\sf S}[i+L]{\sf .U}=1$\;
}
{\tt sortPDecrement}({\sf S})\;
\For{$i=0:L-1$}{
$(\mathcal{L}_i)_0^j=({\sf S}[i]{\rm .L},{\sf S}[i]{\rm .U})$\;
}
$\alpha=L$\;
}
}
$\hat{\mathbf{u}}=\mathbf{L}_0$\;
\end{algorithm}

Here, {\sf S} is a structure type array with the size of $2L$. Each element of
{\sf S} has three members: {\sf P}, {\sf L}, and {\sf U}. The function {\tt
  sortPDecrement} sorts the array {\sf S} by the decreasing order of {\sf P}.

\subsection{CA-SCL Decoding Algorithm for Polar Codes}

The CA-SCL decoding algorithm is used for the CRC-concatenated polar codes. The
difference between the CA-SCL \cite{Tal2012} and the SCL decoding algorithms is how to
make the final decision for $\hat{\mathbf{u}}$. If there is at least one path
satisfying the CRC constraint, the most-reliable CRC-valid path is chosen for
$\hat{\mathbf{u}}$. Otherwise, the decision rule of the SCL decoding algorithm
is used for the CA-SCL decoding algorithm. Since now, without being specified, polar codes
mentioned in the following sections are CRC-concatenated polar codes.

\section{$M$-bit Symbol-based Decoding Algorithm for Polar Codes}
\label{sec:SBDecoder}
\subsection{Generalized Symbol-based SC Decoding Algorithm for Polar Codes}
In \cite{ParSC}, only two-bit, four-bit, eight-bit symbol-based decoding
algorithm for polar codes are discussed. Here, a generalized $M$-bit symbol-based
decoding algorithm for polar codes is discussed. Without loss of generality, assume $M$ is a
power of two, i.e.\ $M=2^m( 0\leq m \leq n)$. Define $\mathcal{IM}_j \defeq
\{jM,jM+1,\cdots,jM+M-1\}\subset \mathcal{I}$, for $0 \leq j <
\frac{N}{M}$. $\mathcal{AM}_j$ and $\mathcal{AM}_j^c$ are defined as:
\begin{equation}
\mathcal{AM}_j \defeq \mathcal{IM}_j \cap \mathcal{A}\hspace{3mm}{\rm and}\hspace{3mm} \mathcal{AM}_j^c \defeq \mathcal{IM}_j \cap \mathcal{A}^c.
\end{equation}

Then the decision rule of the $M$-bit symbol-based SC decoding algorithm can be
described as,
\begin{equation}
\label{equ:MB_DR}
\hat{u}_{jM}^{jM+M-1} =
\underset{\substack{
u_{\mathcal{AM}_j}\in\{0,1\}^{\lvert\mathcal{AM}_j\rvert} \\
u_{{\mathcal{AM}_j}^c}\in\{0\}^{\lvert\mathcal{AM}_j^c\rvert}}}{\arg\max}{\rm P}(\mathbf{y},\hat{u}_0^{jM-1}|u_{jM}^{jM+M-1}),
\end{equation}
where $\lvert\mathcal{AM}_j\rvert$ represents the cardinality of
$\mathcal{AM}_j$. If $M=N$, this
decoding algorithm is a maximum-likelihood sequence decoding algorithm.

If all bits of $\mathbf{u}$ are independent and each bit has an equal probability
of being a 0 or 1, the following symbol-based recursive channel combination relationship can be used to calculate the symbol-based channel transition probability ${\rm P}(\mathbf{y},u_0^{jM-1}|u_{jM}^{jM+M-1})$: 
\begin{PROP}
For any $0\leq m \leq n$, $N=2^n$, $M=2^m$, $0 \leq j < \frac{N}{M}$,
assume $v_0^{N-1}\defeq u_{0,e}^{2N-1}\oplus u_{0,o}^{2N-1}$ and
$v_N^{2N-1}\defeq u_{0,o}^{2N-1}$, then
\begin{equation}
\label{eq:prop1}
\begin{split}
{\rm P}(y_{0}^{2N-1},&u_0^{2jM-1}|u_{2jM}^{2jM+2M-1}) = \\
&{\rm P}(y_0^{N-1},v_0^{jM-1}|v_{jM}^{jM+M-1})\\
&\cdot{\rm P}(y_{N}^{2N-1},v_N^{N+jM-1}|v_{N+jM}^{N+jM+M-1})
\end{split}
\end{equation}
\end{PROP}

\begin{proof}
According to Bayes' theorem, 
\begin{equation}
\label{eq:prop2}
\begin{split}
{\rm P}(y_0^{2N-1},&u_0^{2jM-1}|u_{2jM}^{2jM+2M-1})\\
&= \frac{{\rm
  P}(y_0^{2N-1},u_0^{2jM+2M-2}|u_{2jM+2M-1})}{{\rm
  P}(u_{2jM}^{2jM+2M-2}|u_{2jM+2M-1})}.
\end{split}
\end{equation}
Because all bits of $\mathbf{u}$ are independent and each bit has an equal
probability of being a 0 or 1,
\begin{equation*}
\begin{split}
{\rm  P}&(u_{2jM}^{2jM+2M-2}|u_{2jM+2M-1}) = {\rm  P}(u_{2jM}^{2jM+2M-2}) \\
&={\rm  P}(u_{2jM}){\rm  P}(u_{2jM})\cdots{\rm  P}(u_{2jM+2M-2})=2^{-(2M-1)}.
\end{split}
\end{equation*}
Therefore,
\begin{equation}
\label{eq:prop3}
\begin{split}
{\rm P}(y&_0^{2N-1},u_0^{2jM-1}|u_{2jM}^{2jM+2M-1})\\
&= 2^{(2M-1)}{\rm P}(y_0^{2N-1},u_0^{2jM+2M-2}|u_{2jM+2M-1}).
\end{split}
\end{equation}
According to Eq.~\eqref{eq:ArikanT2},
\begin{equation}
\label{eq:prop4}
\begin{split}
{\rm P}(y_0^{2N-1},&u_0^{2jM+2M-2}|u_{2jM+2M-1}) \\
&=\frac{1}{2}{\rm P}(y_0^{N-1},v_0^{jM+M-2}|v_{jM+M-1})\\
&\cdot{\rm P}(y_N^{2N-1},v_N^{N+jM+M-2}|v_{N+jM+M-1}).
\end{split}
\end{equation}

According to the definition of $v_{0}^{N-1}$, all bits of $v_{0}^{N-1}$ are
independent and ${\rm P}(v_j=0)={\rm P}(v_j=1)=\frac{1}{2}$ for $0 \leq j < N$.
Then we have
\begin{equation}
\label{eq:prop5}
\begin{split}
{\rm P}(y_0^{N-1},&v_0^{jM+M-2}|v_{jM+M-1}) \\
&= 2^{-(M-1)}{\rm P}(y_0^{N-1},v_0^{jM-1}|v_{jM}^{jM+M-1}).
\end{split}
\end{equation}
Similarly,
\begin{equation}
\label{eq:prop6}
\begin{split}
{\rm P}(y&_N^{2N-1},v_N^{N+jM+M-2}|v_{N+jM+M-1})\\
&= 2^{-(M-1)}{\rm P}(y_N^{2N-1},v_N^{N+jM-1}|v_{N+jM}^{N+jM+M-1}).
\end{split}
\end{equation}
Then, by equations ~\eqref{eq:prop3} $\sim$ \eqref{eq:prop6},
Eq.~\eqref{eq:prop1} is obtained.
\end{proof}

Similar to the SC decoding algorithm, an $M$-bit symbol-based SC decoding
algorithm can be represented by using a message flow graph (MFG) as well, where
a channel transition probability is referred as a {\it message} for the sake of
convenience. If the code length of a polar code is $N$, the MFG can be divided
into $(n+1)$ stages ${\rm S}_0, {\rm S}_1, 
\cdots, {\rm S}_n$: one initial stage ${\rm S}_0$ and $n$ calculation stages. For the SC
decoding algorithm, all calculation stages carry out the calculation of Eq.~\eqref{eq:ArikanT1}
and \eqref{eq:ArikanT2}. However, for the $M$-bit symbol-based SC decoding
algorithm, the Arikan's recursive transformations are performed in the first
$(n-m)$ calculation stages, called channel transformation stages. In the last $m$ calculation
stages, called channel combination stages, Eq.~\eqref{eq:prop1} is used to compute
messages. Therefore, an $M$-bit symbol-based SC decoding algorithm contains two 
parts. The first part contains calculations of the first $n-m$ stages and
consists of $M$ SC decoders for polar codes of length $\frac{N}{M}$. These SC
decoders are called as component decoders. There are no message exchange between
these component decoders. Channel combination stages use outputs of channel
transformation stages to calculate symbol-based messages and feed the
estimated symbol back to component decoders to update partial-sums.

For example, as shown inFig.~\ref{fig:MFG_8_4}, the MFG of a
four-bit symbol-based SC decoding algorithm for a polar code with $L=8$  has four stages. Messages of the initial stage (S0) come
from the channel directly. Messages of the first stage (S1) are
calculated with Arikan's recursive transformations. Messages of the second
and third stages (S2 and S3) are calculated with Eq.~\eqref{eq:prop1}. The four
small gray boxes on the right are four component SC decoders. And stages in the
big gray box on the left are channel combination stages. Here, 
\begin{align*}
v_0^3&=u_{0,e}^7\oplus u_{0,o}^7, \hspace{4mm} v_4^7=u_{0,o}^7,\\
w_0&=v_0\oplus v_1=u_0\oplus u_1 \oplus u_2 \oplus u_3,\\
w_1&=v_2\oplus v_3=u_4\oplus u_5 \oplus u_6 \oplus u_7,\\
w_2&=v_1=u_2\oplus u_3,\\
w_3&=v_3=u_6\oplus u_7,\\
w_4&=v_4\oplus v_5 = u_1\oplus u_3,\\
w_5&=v_6\oplus v_7 = u_5 \oplus u_7,\\
w_6&=v_5 = u_3,\\
w_7&=v_7 = u_7. 
\end{align*}

\begin{figure}[htbp]
\centering
\includegraphics[width=6cm]{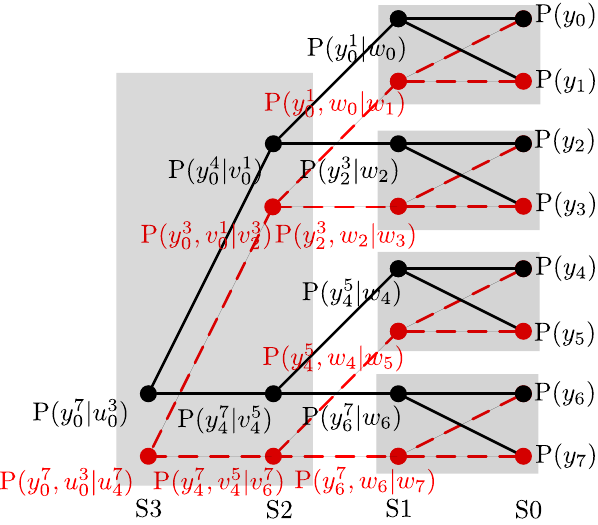}
\caption{The message flow graph of a four-bit symbol-based SC decoding algorithm
  for a polar code with a length of eight.}
\label{fig:MFG_8_4}
\end{figure}

We can take advantage of the symbol-based channel combination to reduce
complexity of calculating the symbol-based channel transition probability. In
\cite{ParSC}, an ML decoder is use to calculate the symbol-based message of 
stage ${\rm S}_{n}$ from output of component decoders directly. There are $2^M$ possible values for an $M$-bit
symbol. \cite{6464502} shows that $(M-1)$ additions are needed to calculate the
log-likelihood (LL) message corresponding to each
value. Therefore, an ML decoder needs $2^M(M-1)$ additions in total. 
In channel combination stages, there are $2^{n-i}$ nodes in the $i$-th stage and each
node contains $2^{M+i-n}$ messages. One addition is
needed to compute each LL message according to
Eq.~\eqref{eq:prop1}. Hence, channel combination stages need
$\sum_{i=0}^{m-1}2^{i}2^{\frac{M}{2^{i}}}$ additions in total. For the example
shown in Fig.~\ref{fig:MFG_8_4}, the ML decoder needs $2^4(4-1)=48$
additions. The channel combination stages need only $2^4+2\times 2^2=24$
additions, which is only a half of those needed by the ML decoder.  

In terms of the error performance, simulations of
\cite{ParSC} show that there is no observed performance loss for the the $M$-bit symbol-based SC decoding
algorithm using the ML decoder to calculate the symbol-based message, compared
with the SC polar decoding algorithm. Since our
channel combination relationship can be used to provide the same calculation results as the ML decoder used in
\cite{ParSC} does, 
the $M$-bit symbol-based SC decoding algorithm using the
symbol-based channel combination relationship does not have any observed
performance degradation compared with the SC decoding algorithm.

\subsection{Generalized Symbol-based SCL Decoding Algorithm for Polar Codes}
The symbol-based SCL decoding algorithm is more complex than the SCL algorithm, since
the path expansion coefficient is not a constant any more. In the SCL algorithm, for each
information bit, the path expansion coefficient is two. But for the $M$-bit
symbol-based SCL decoding algorithm, the path expansion coefficient is
$2^{\lvert\mathcal{AM}_j\rvert}$, which depends on
the number of information bits in an $M$-bit symbol. The $M$-bit symbol-based
SCL decoding algorithm is described in Alg.~\ref{alg:SBSCL}.

\begin{algorithm}
\caption{$M$-bit Symbol-based SCL Decoding Algorithm \cite{ParSC}}
\label{alg:SBSCL}
\LinesNumbered
$\alpha=1$\;
\For{$j=0:\frac{N}{M}-1$}{
$\beta=2^{\lvert\mathcal{AM}_j\rvert}$\;
\uIf{$\beta==1$}{
\For{$i=0:\alpha-1$}{
$(\mathcal{L}_i)_{jM}^{jM+M-1} = \mathbf{0}$\;
}
}
\uElseIf{$\alpha\beta\leq L$}{
$u_{\mathcal{AM}_j^c}=\mathbf{0}$\;
\For{$k=0:\beta-1$}{
$u_{\mathcal{AM}_j}=${\tt dec2bin}$(k,\lvert\mathcal{AM}_j\rvert)$\;
\For{$i=0:\alpha-1$}{
$t=i+k\alpha$\;
$(\mathcal{L}_{t})_{0}^{jM+M-1} = ((\mathcal{L}_i)_{0}^{jM-1},u_{jM}^{jM+M-1})$\;
}
}
$\alpha=\alpha\beta$\;
}
\Else{
$u_{\mathcal{AM}_j^c}=\mathbf{0}$\;
\For{$k=0:\beta-1$}{
$u_{\mathcal{AM}_j}=${\tt dec2bin}$(k,\lvert\mathcal{AM}_j\rvert)$\;
\For{$i=0:L-1$}{
$t=i+kL$\;
${\sf S}[t]{\sf .P}={\rm P}(\mathbf{y},(\mathcal{L}_i)_0^{jM-1}|u_{jM}^{jM+M-1})$\;
${\sf S}[t]{\sf .L}=(\mathcal{L}_i)_0^{jM-1}$\;
${\sf S}[t]{\sf .U}=u_{jM}^{jM+M-1}$\;
}
}
{\tt sortPDecrement}({\sf S})\;
\For{$i=0:L-1$}{
$(\mathcal{L}_i)_0^{jM+M-1}=({\sf S}[i]{\rm .L},{\sf S}[i]{\rm .U})$\;
}
$\alpha=L$\;
}
}
\end{algorithm}

Here, without any ambiguity, $\mathbf{0}$ represents a zero vector whose bit-width
is determined by the left-hand operator. The function {\tt dec2bin}$(d,b)$
converts a decimal number $d$ to a $b$-bit binary vector. Eq.~\eqref{eq:prop1}
can also be used to calculate the symbol-based channel transition
probability corresponding to each list, i.e.\ ${\rm P}(\mathbf{y},(\mathcal{L}_i)_0^{jM-1}|u_{jM}^{jM+M-1})$.

\begin{figure}[htbp]
\centering
\includegraphics[width=7.5cm]{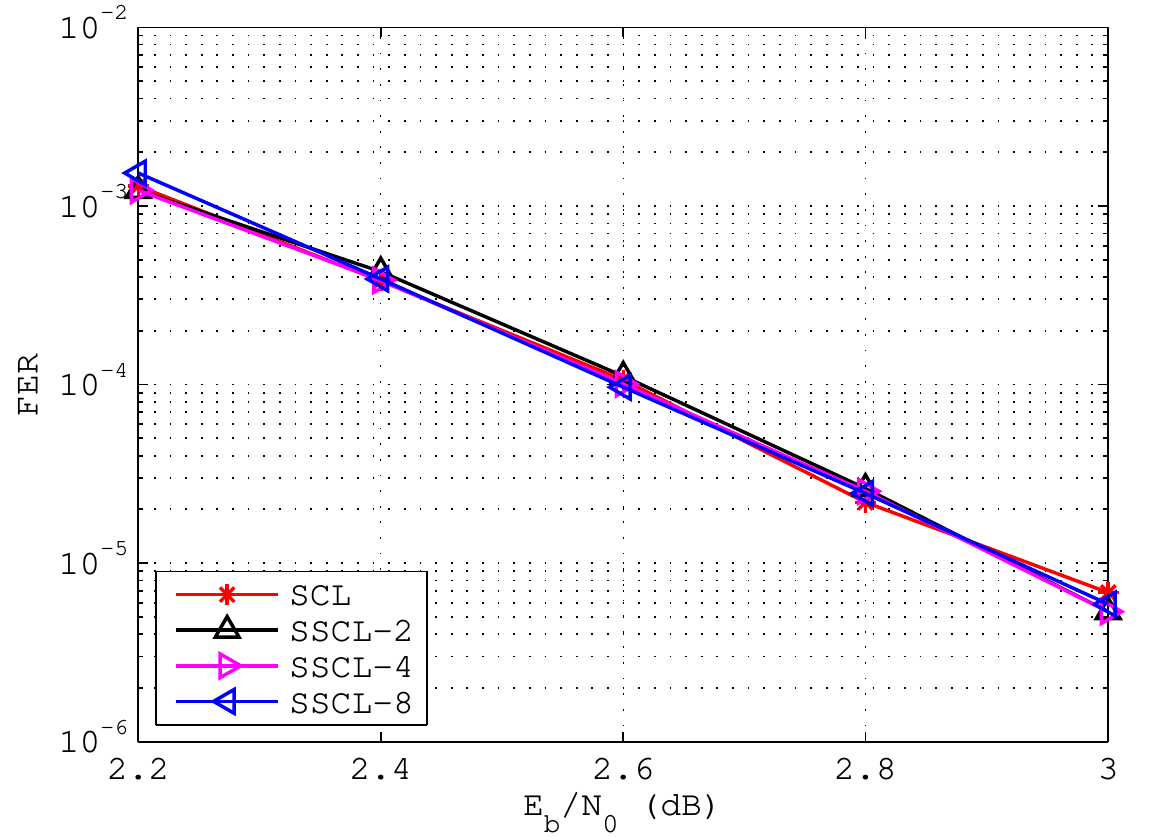}
\caption{FERs of symbol-based SCL decoding algorithms for a
  (1024, 512) polar code with $L=4$.}
\label{fig:SBSCL_PER}
\end{figure}

Fig.~\ref{fig:SBSCL_PER} shows FERs of symbol-based SCL
decoding algorithms for a (1024,512) polar code with $L=4$. Performance differences
between these curves are very minor. Therefore, by applying Eq.~\eqref{eq:prop1} 
the symbol-based SCL algorithm does not introduce the obvious performance loss
compared with the SCL decoding algorithm. Even with different $M$s, these
performance curves are very close. Here, SSCL-$i$ denotes the $i$-bit symbol-based
SCL decoding algorithm.

\section{Two-Stage List Pruning Network}
\label{sec:TSLPN}
For the $M$-bit symbol-based SCL decoding algorithm, the maximum path expansion
coefficient is $2^M$, i.e.\, each existing list generates $2^M$ list
candidates. Therefore, in the worst-case scenario, $L$ most-reliable lists should be sorted 
out of $2^ML$ list candidates. To facilitate this sorting network, we propose
a two-stage list pruning network. In the first stage, $q$ most-reliable lists
are found
out among $2^M$ list candidates of each existing list. Therefore, there are $qL$ list
candidates left. In the second stage, the $L$ most-reliable lists are sorted out
from the $qL$ list candidates generate by the first stage. The message flow of a two-stage list pruning
network is illustrated in Fig.~\ref{fig:TSLPP}. 
 
\begin{figure}[htbp]
\centering
\includegraphics[width=7.5cm]{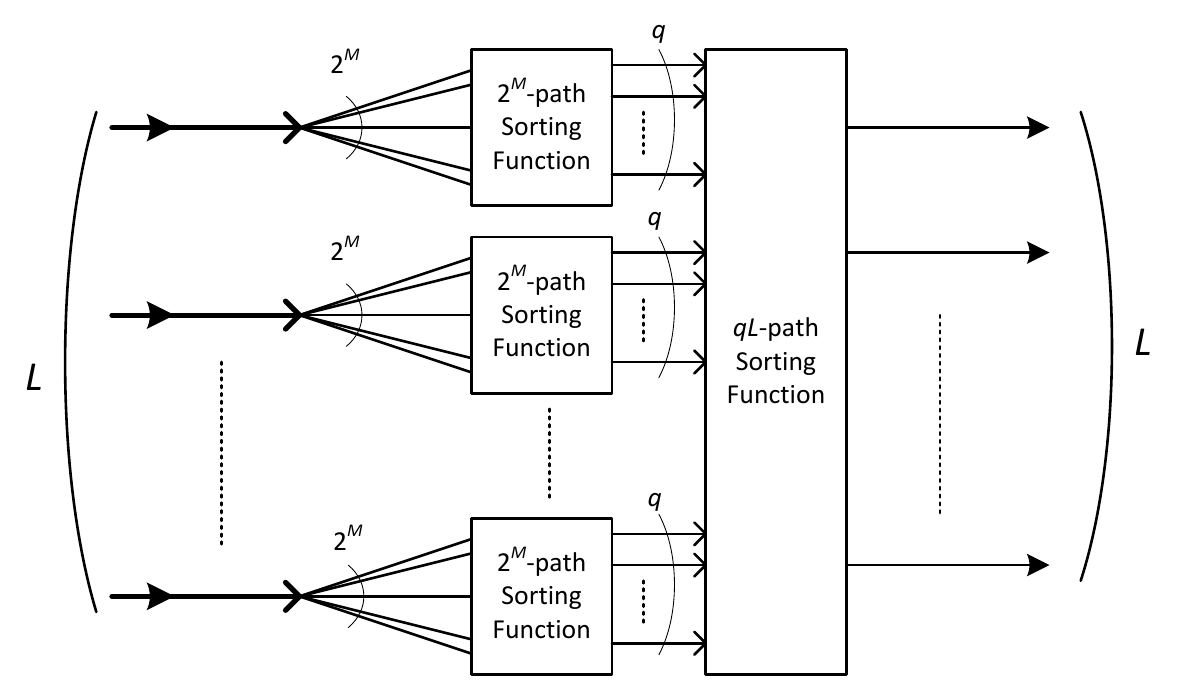}
\caption{Message flow for a two-stage list pruning network.}
\label{fig:TSLPP}
\end{figure}

It is easy to prove that if $q \geq L$ and $2^M > L$, the $L$ lists found by the two-stage list
pruning network are exactly the $L$ most-reliable lists among the $2^ML$ list
candidates. Therefore, we only consider $q \leq L$. In terms of complexity, a smaller $q$ leads to a two-stage list pruning
network with lower complexity but the probability that the $L$ lists found by the two-stage list
pruning network are exactly the $L$ most-reliable lists among the $2^ML$ list
candidates decreases as well. This may cause some performance loss. 

Fig.~\ref{fig:TSLPP_4} and \ref{fig:TSLPP_16}
show how different $q$s affect FERs of an SSCL-8 decoding algorithm for a (1024, 512) polar code with $L=4$ and $L=16$, respectively. When $L=4$ and $q=2$, the SSCL-8 decoding algorithm
shows an FER performance
loss of about 0.2 dB. When $L=16$, compared with the
FER performance with $q=16$, there is no observed performance degradation when $q=8$. The
performance loss due to $q=4$ is about 0.08 dB. Therefore, for 
$L=16$, to reduce complexity and the latency of the two-stage list pruning network,
$q$ can be $8$. If the 0.08 dB performance loss is tolerated, $q$ can be reduced further to four.

\begin{figure}[htbp]
\centering
\includegraphics[width=7.5cm]{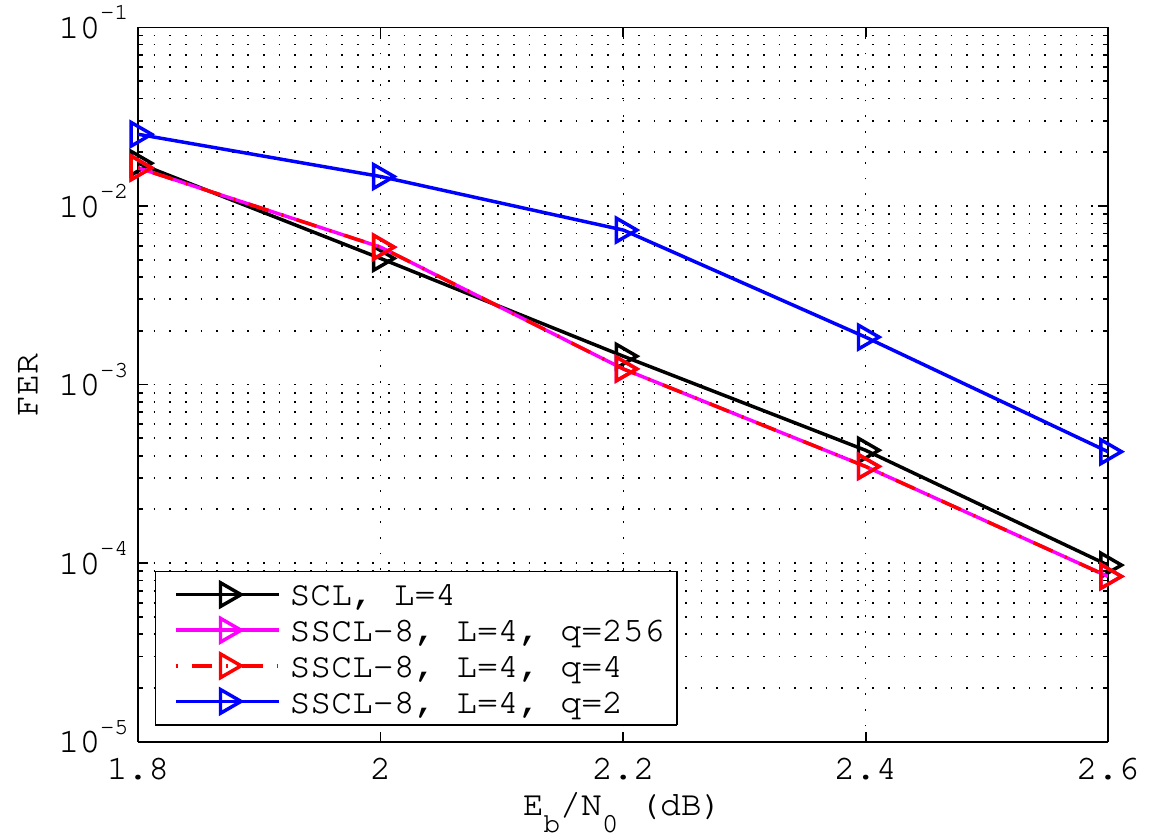}
\caption{FERs of SSCL-8 decoder for a
  (1024, 512) polar code with $L=4$.}
\label{fig:TSLPP_4}
\end{figure}


\begin{figure}[htbp]
\centering
\includegraphics[width=7.5cm]{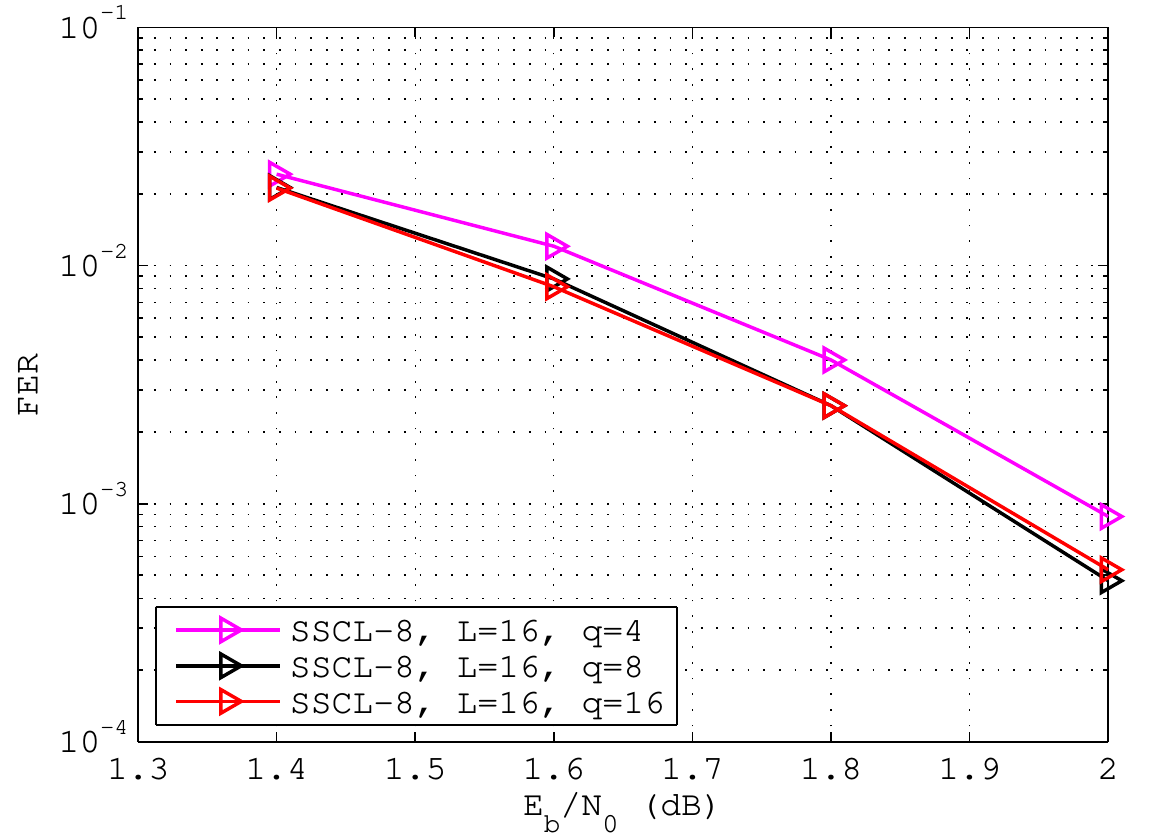}
\caption{FERs of SSCL-8 decoder for a
  (1024, 512) polar code with $L=16$.}
\label{fig:TSLPP_16}
\end{figure}

Similarly, as shown in Fig.~\ref{fig:TSLPP_8_2048}, for a (2048,1433) polar
code, the two stage list-pruning network of $q=4$ helps to reduce 
the complexity of SSCL-8 decoder without the obvious performance loss.

\begin{figure}[htbp]
\centering
\includegraphics[width=7.5cm]{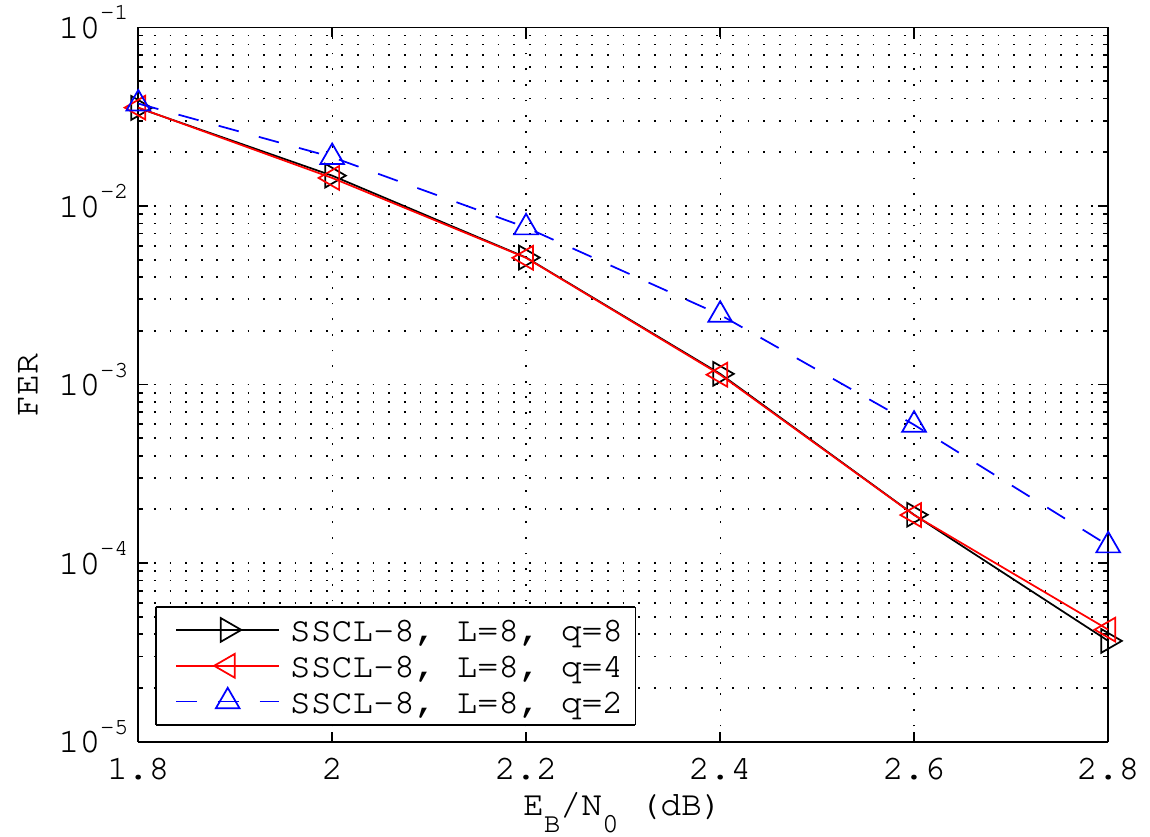}
\caption{FERs of SSCL-8 decoding algorithm for a
  (2048, 1433) polar code with $L=8$.}
\label{fig:TSLPP_8_2048}
\end{figure}

Fig.~\ref{fig:TSLPP_P4_L8} shows FERs of an SSCL-4 decoder for a (1024,512) polar code with $L=8$ while
different $q$s are used. Compared with the case of $q=8$, there is no obvious FER
performance loss when $q=4$. However, $q=2$ incurs an FER performance loss of about
0.3 dB when the FER is $10^{-3}$.

\begin{figure}[htbp]
\centering
\includegraphics[width=7.5cm]{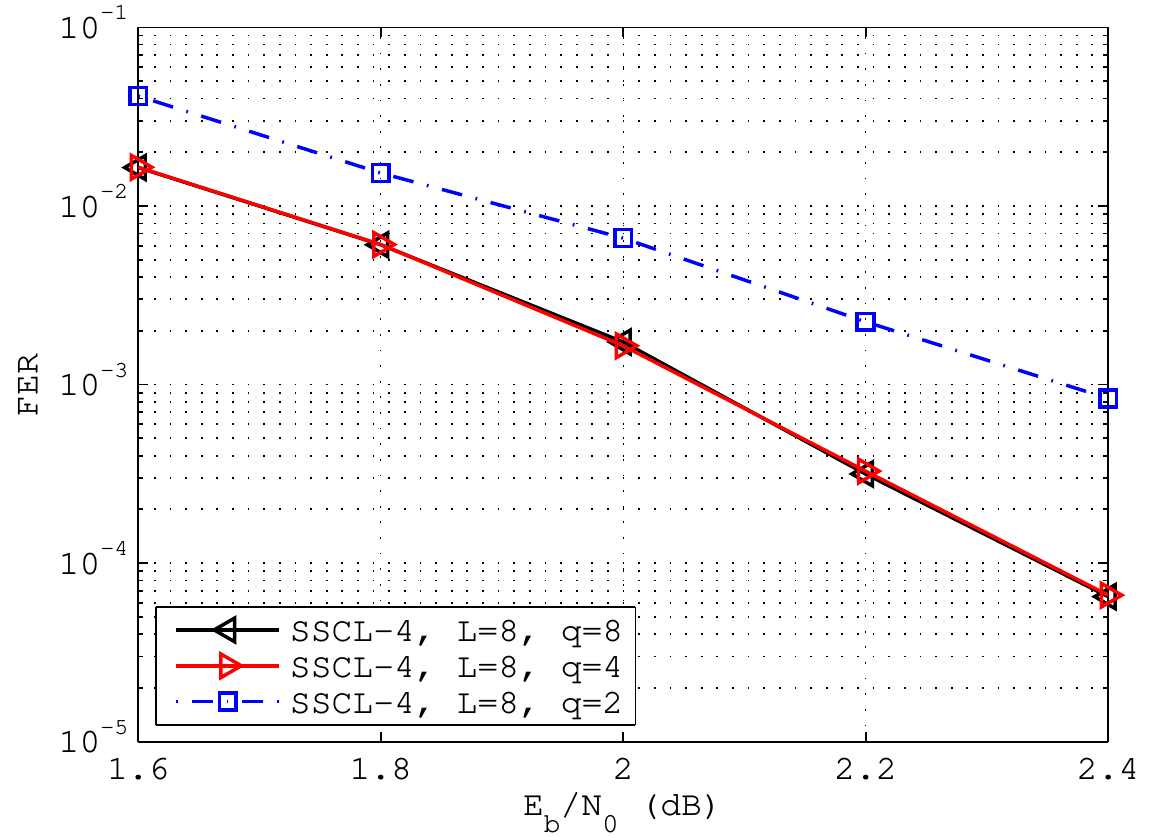}
\caption{FERs of SSCL-4 decoder for a
  (1024, 512) polar code with $L=8$.}
\label{fig:TSLPP_P4_L8}
\end{figure}

To illustrate advantages of two-stage list pruning network, two tree
sorting networks are designed to find the 8 maximal values out of 128 values which can be
used in the SSCL-4 decoder with $L=8$. One is a
conventional tree sorting network, shown in Fig.~\ref{fig:STP4L8}, refered to as
CTSN. The other is
a two-stage tree sorting network with $q=4$, shown in
Fig.~\ref{fig:STP4L8TS}, refered to as TSTSN. Here, the "ps16to8" block is a bitonic sorter which
finds the maximal 8 values out of 16 values. The bitonic sorter, "ps8to4", finds
the maximal 4 values out of 8 values.
\begin{figure}[htbp]
\centering
\includegraphics[width=7cm]{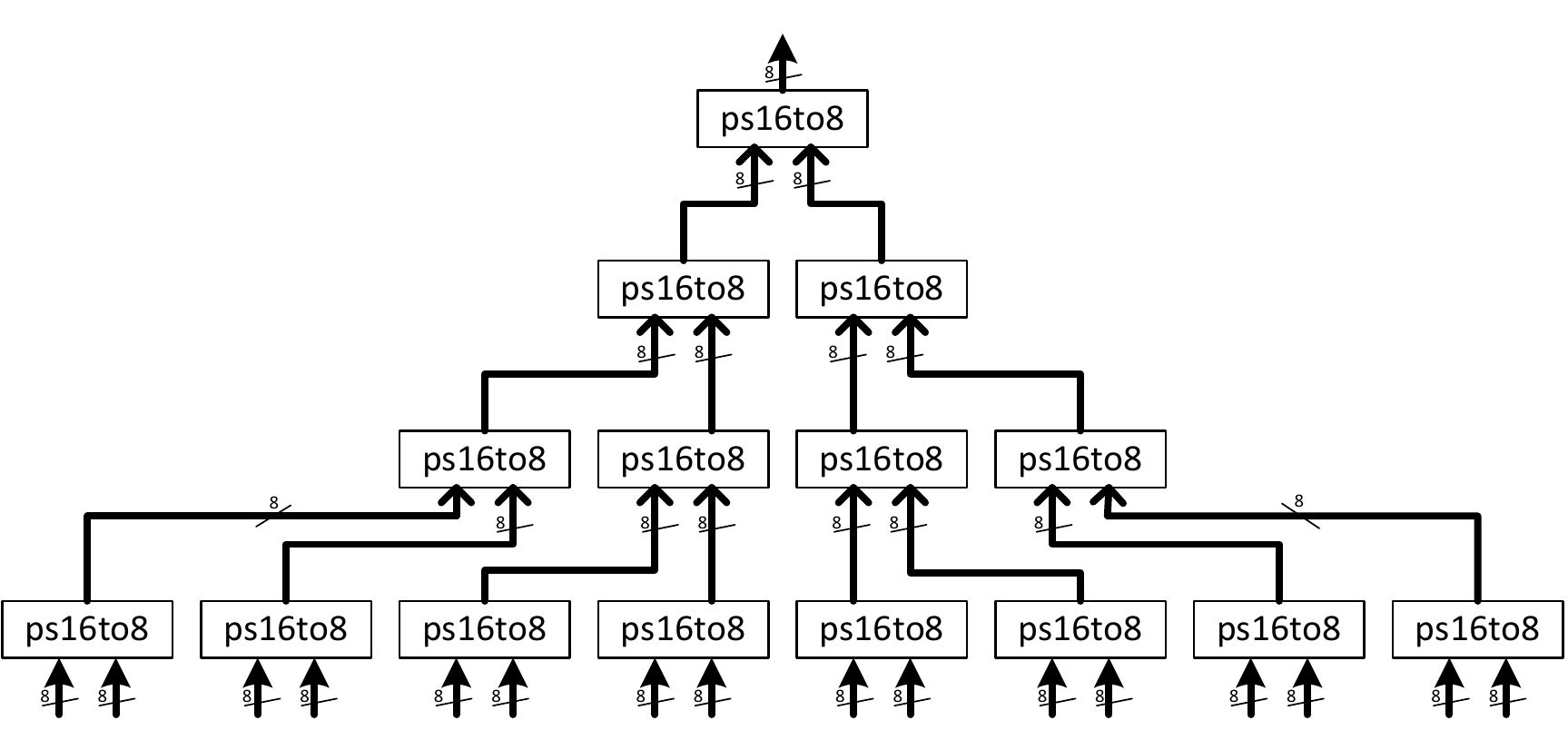}
\caption{A conventional tree sorting network to find the 8 maximal values out of 128 values.}
\label{fig:STP4L8}
\end{figure}

\begin{figure}[htbp]
\centering
\includegraphics[width=7cm]{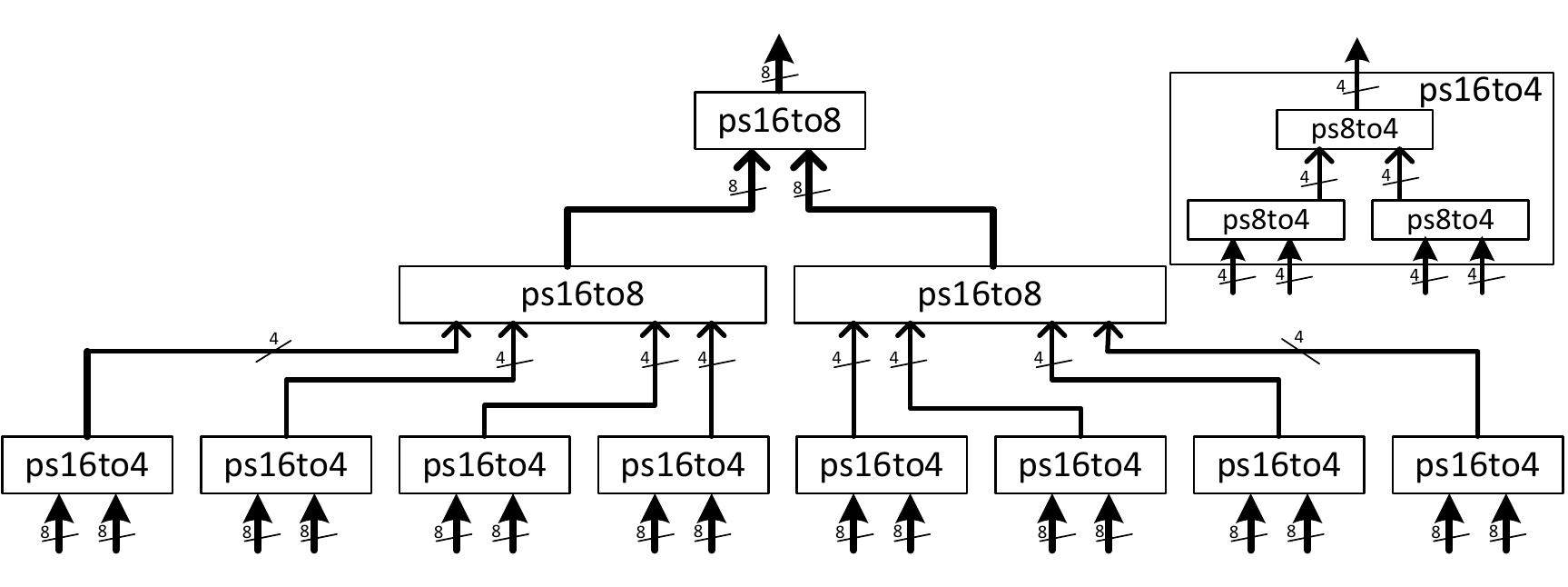}
\caption{A two-stage tree sorting network to find the 8 maximal values out of
  128 values.}
\label{fig:STP4L8TS}
\end{figure}

We implement these two sorting networks and use the RTL compiler to synthesize
them with a TSMC 90-nm CMOS technology. The TSTSN has a smaller area and a shorter
critical path than the CTSN, as shown in Table~\ref{tab:TS_syn_result}. Besides, the TSTSN does not introduce any obvious
performance degradation as shown in Fig.~\ref{fig:TSLPP_P4_L8}.

\begin{table}[hbtp]
\begin{center}
\caption{Synthesizing results for CTSN and TSTSN.}
\label{tab:TS_syn_result}
\begin{tabular}{|c|c|c|}
\hline
Design & area ($\text{mm}^2$) & Critical Path Delay (ns) \\ \hline
CTSN & 0.206 & 7.463 \\ \hline
TSTSN & 0.134 & 5.861 \\ \hline
\end{tabular}
\end{center}
\end{table}

\section{Conclusion}
\label{sec:conclusion}
In this paper, we discuss the generalized symbol-based SC and SCL decoding algorithm for
polar codes and derive the recursive procedure to calculate the symbol-based
channel transition probability. This recursive procedure needs less
additions than the ML scheme used in \cite{ParSC}. A two-stage list pruning
network is also proposed to simplify the $L$-list finding network. 

%
%



\bibliographystyle{IEEEtran}

\begin{thebibliography}{10}
\providecommand{\url}[1]{#1}
\csname url@samestyle\endcsname
\providecommand{\newblock}{\relax}
\providecommand{\bibinfo}[2]{#2}
\providecommand{\BIBentrySTDinterwordspacing}{\spaceskip=0pt\relax}
\providecommand{\BIBentryALTinterwordstretchfactor}{4}
\providecommand{\BIBentryALTinterwordspacing}{\spaceskip=\fontdimen2\font plus
\BIBentryALTinterwordstretchfactor\fontdimen3\font minus
  \fontdimen4\font\relax}
\providecommand{\BIBforeignlanguage}[2]{{%
\expandafter\ifx\csname l@#1\endcsname\relax
\typeout{** WARNING: IEEEtran.bst: No hyphenation pattern has been}%
\typeout{** loaded for the language `#1'. Using the pattern for}%
\typeout{** the default language instead.}%
\else
\language=\csname l@#1\endcsname
\fi
#2}}
\providecommand{\BIBdecl}{\relax}
\BIBdecl

\bibitem{5075875}
E.~Arikan, ``Channel polarization: A method for constructing capacity-achieving
  codes for symmetric binary-input memoryless channels,'' \emph{{IEEE} Trans.
  Inf. Theory}, vol.~55, no.~7, pp. 3051--3073, July 2009.

\bibitem{5351487}
E.~Sasoglu, I.~Telatar, and E.~Arikan, ``Polarization for arbitrary discrete
  memoryless channels,'' in \emph{ITW}, Oct 2009, pp. 144--148.

\bibitem{6327689}
C.~Leroux, A.~Raymond, G.~Sarkis, and W.~Gross, ``A semi-parallel
  successive-cancellation decoder for polar codes,'' \emph{{IEEE} Trans. Signal
  Process.}, vol.~61, no.~2, pp. 289--299, Jan 2013.

\bibitem{6297420}
K.~Niu and K.~Chen, ``{CRC}-aided decoding of polar codes,'' \emph{{IEEE}
  Commun. Lett.}, vol.~16, no.~10, pp. 1668--1671, October 2012.

\bibitem{6033837}
A.~Eslami and H.~Pishro-Nik, ``A practical approach to polar codes,'' in
  \emph{ISIT}, Jul. 2011, pp. 16--20.

\bibitem{5934670}
E.~Arikan, ``Systematic polar coding,'' \emph{{IEEE} Commun. Lett.}, vol.~15,
  no.~8, pp. 860--862, August 2011.

\bibitem{6033904}
I.~Tal and A.~Vardy, ``List decoding of polar codes,'' in \emph{ISIT}, July
  2011, pp. 1--5.

\bibitem{Tal2012}
------, ``List decoding of polar codes,'' arXiv:1206.0050, Jun. 2012.

\bibitem{1603394}
\emph{IEEE Standard for Local and Metropolitan Area Networks Part 16: Air
  Interface for Fixed and Mobile Broadband Wireless Access Systems Amendment 2:
  Physical and Medium Access Control Layers for Combined Fixed and Mobile
  Operation in Licensed Bands and Corrigendum 1}, IEEE Std. 802.16e-2005, Mar.
  2006.

\bibitem{5946819}
C.~Leroux, I.~Tal, A.~Vardy, and W.~Gross, ``Hardware architectures for
  successive cancellation decoding of polar codes,'' in \emph{ICASSP}, May
  2011, pp. 1665--1668.

\bibitem{6065237}
A.~Alamdar-Yazdi and F.~Kschischang, ``A simplified successive-cancellation
  decoder for polar codes,'' \emph{{IEEE} Commun. Lett.}, vol.~15, no.~12, pp.
  1378--1380, Dec. 2011.

\bibitem{6680761}
C.~Zhang and K.~Parhi, ``Latency analysis and architecture design of simplified
  sc polar decoders,'' \emph{{IEEE} Trans. Circuits Syst. {II}}, vol.~61,
  no.~2, pp. 115--119, Feb. 2014.

\bibitem{6464502}
G.~Sarkis and W.~Gross, ``Increasing the throughput of polar decoders,''
  \emph{{IEEE} Commun. Lett.}, vol.~17, no.~4, pp. 725--728, Apr. 2013.

\bibitem{6804939}
G.~Sarkis, P.~Giard, A.~Vardy, C.~Thibeault, and W.~Gross, ``Fast polar
  decoders: Algorithm and implementation,'' \emph{IEEE Journal on Selected
  Areas in Communications}, vol.~32, no.~5, pp. 946--957, May 2014.

\bibitem{6475198}
C.~Zhang and K.~Parhi, ``Low-latency sequential and overlapped architectures
  for successive cancellation polar decoder,'' \emph{{IEEE} Trans. Signal
  Process.}, vol.~61, no.~10, pp. 2429--2441, May 2013.

\bibitem{ListPolarJun1}
J.~Lin and Z.~Yan, ``Efficient list decoder architecture for polar codes,'' in
  \emph{ISCAS}, 2014, to appear.

\bibitem{ParSC}
B.~Li, H.~Shen, and D.~Tse, ``Parallel decoders of polar codes,''
  arXiv:1309.1026, September 2013.

\end{thebibliography}

%
%
%

\end{document}